 \documentclass[aps,pra,reprint,superscriptaddress,longbibliography,floatfix,twocolumn, nofootinbib,notitlepage]{revtex4-2} 
\pdfoutput=1 

\usepackage{dsfont} 	
\usepackage{microtype} 	
\usepackage[utf8]{inputenc} 	
\usepackage[T1]{fontenc} 	
\usepackage[british]{babel} 	
\usepackage[dvipsnames,x11names]{xcolor} 	
\usepackage{amsmath,amssymb,amsthm,bm,amsfonts,bbm} 	
\usepackage{subcaption}
\usepackage{mathtools} 	
\usepackage{physics} 	
\usepackage[colorlinks=true,citecolor=green,linkcolor=Purple,urlcolor=Magenta]{hyperref} 
\usepackage{verbatim} 
\usepackage{empheq}
\usepackage{hhline}
\usepackage{colortbl}
\usepackage[theorems, breakable]{tcolorbox}
\usepackage[normalem]{ulem}

\usepackage{caption}
\newcommand{\id}{\mathds{1}}

\renewcommand{\L}{\mathcal{L}}

\makeatletter 
\newcommand{\doublewidetilde}[1]{{%
  \mathpalette\double@widetilde{#1}%
}}
\newcommand{\double@widetilde}[2]{%
  \sbox\z@{$\m@th#1\widetilde{#2}$}%
  \ht\z@=.9\ht\z@
  \widetilde{\box\z@}%
}
  
\makeatother

\newtheorem{definition}{Definition}
\newtheorem{theorem}{Theorem}
\newtheorem{lemma}{Lemma}
\newtheorem{corollary}{Corollary}

\definecolor{brow}{rgb}{0.8515625,0.67578125,0.5234375} 
\definecolor{blu}{rgb}{0.39453125,0.6171875,0.734375} 

\hypersetup{pdftitle={{Characterising transformation between arbitrary quantum objects and the emergence of quantum indefinite causality}}} 
\begin{document}
\setlength\fboxrule{1pt}
\title{Certifying measurement incompatibility in prepare-and-measure and Bell scenarios}

\author{Sophie Egelhaaf}
\affiliation{Department of Applied Physics, University of Geneva, 1211 Geneva, Switzerland}
\author{Jef Pauwels}
\affiliation{Department of Applied Physics, University of Geneva, 1211 Geneva, Switzerland}
\affiliation{Constructor University, Geneva, Switzerland}
\author{Marco Túlio Quintino}
\affiliation{{Sorbonne Universit\' {e}, CNRS, LIP6, F-75005 Paris, France}}
\author{Roope Uola}
\affiliation{Department of Applied Physics, University of Geneva, 1211 Geneva, Switzerland}

\begin{abstract}
We consider the problem of certifying measurement incompatibility in a prepare-and-measure (PM) scenario. We present different families of sets of qubit measurements which are incompatible, but cannot lead to any quantum over classical advantage in PM scenarios. Our examples are obtained via a general theorem which proves a set of qubit dichotomic measurements can have their incompatibility certified in a PM scenario if and only if their incompatibility can be certified in a bipartite Bell scenario where the parties share a maximally entangled state. Our framework naturally suggests a hierarchy of increasingly stronger notions of incompatibility, in which more power is given to the classical simulation by increasing its dimensionality. For qubits, we give an example of measurements whose incompatibility can be certified against trit simulations, which we show is the strongest possible notion for qubits in this framework. 

\end{abstract}
\maketitle
\section{Introduction}
Measurement incompatibility, the existence of measurements that cannot be carried out simultaneously on a single system, is a defining feature of quantum theory. It represents a radical departure from classical theory and underpins central aspects in quantum foundations, such as the resolution of the Heisenberg microscope \cite{buschrmp2014}, fundamental bounds in interferometry \cite{Kiukas22coherence}, and Fine's theorem \cite{fine82}. More recently, measurement incompatibility has found various uses in quantum information theory \cite{Guhne_JMreview} as a key ingredient for quantum information protocols, such as quantum key distribution~\cite{masini2024joint,roydeloison2024deviceindependent}, channel verification \cite{Pusey2015}, device-independent entanglement certification \cite{Lobo2023}, EPR steering \cite{wiseman07,dani_paul_review,Uola_Steerreview,quintino14,uola14,uola15,kiukas17}, Bell nonlocality \cite{brunner_review,wolf09}, temporal correlations \cite{Clemente2015,uola19a,uola2022retrievability}, Random Access Coding \cite{Carmeli2020}, state discrimination tasks \cite{skrzypczyk19,carmeli19a,oszmaniec19,uola19b,uola19c}, distributed sampling \cite{guerini2019}, contextuality \cite{tavakoli20,selby21}, and programmability \cite{buscemi20}.

Incompatibility is defined at the abstract level of Hilbert spaces, as a relational property between operator measures. In practice, however, neither these abstract objects nor the relations between them, are directly observable. We only have access to the outcome statistics. While incompatibility can be formulated operationally, its direct detection hinges on the assumption that appropriately chosen sets of trusted test-states can be prepared.

Another feature of quantum theory is the ability to establish correlations between distant parties that defy classical explanations either by measuring quantum messages, as in the prepare-and-measure scenario \cite{ambainis1998_RAC,ambainis2008_RAC,galleo2010gallego,DallArno2022}, or by performing local measurements on pre-shared entangled states, as in the Bell scenario \cite{bell64,brunner_review}. In the Bell scenario, reaching such non-classical correlations is known to require incompatible measurements and, hence, can be used to certify incompatibility in a device-independent way, a fact studied extensively \cite{Chen2016,Cavalcanti2016,Chen2021,Quintino2019}.

Here, we analyse the problem of certifying measurement incompatibility in a PM scenario. This problem has appeared in some previous works, either implicitly or explicitly~\cite{Frenkel2015_Weiner,viera2022_interplays,degois21_PMJM,tavakoli2018_selt_test,Saha23}. Here, we provide a rigorous definition for the phenomenon, consistent with the previous works, and discuss some important properties. Focusing in particular on the simplest case of qubit measurements, we show that various fundamental results in the literature \cite{divianszky23_PM_Bell_Grothendieck,renner23_PM,Frenkel2015_Weiner,degois21_PMJM} can be linked to our scenario, providing a better structural understanding of the concept of PM certifiable incompatibility, together with classical models and witnessing techniques for the phenomenon. 

The paper is structured as follows. We present basic definitions and discuss their operational role in the second section. In the third section, we summarize the central correlation scenarios: the prepare-and-measure and the Bell scenario. In the fourth section, we illustrate incompatibility in the PM scenario against classical simulation protocols. More precisely, we show that all sets of qubit measurements have a four-dimensional classical model and show that the sets of qubit measurements which have a two-dimensional classical model also have a local model when used in a Bell scenario with a maximally entangled state. We further provide examples of measurement sets illustrating this. Our notion of incompatibility naturally suggests a hierarchy of increasingly stronger notions of incompatibility corresponding to the dimensionality of the underlying classical simulation. For qubits, we show that all measurements can be simulated with quart models, and give an example of incompatible measurements that can be certified incompatible assuming trit models. We conclude the paper in the fifth section with discussion and open questions.

\section{Joint measurability}

\begin{figure}[ht!]
\centering
\includegraphics[width= 0.9\linewidth, clip]{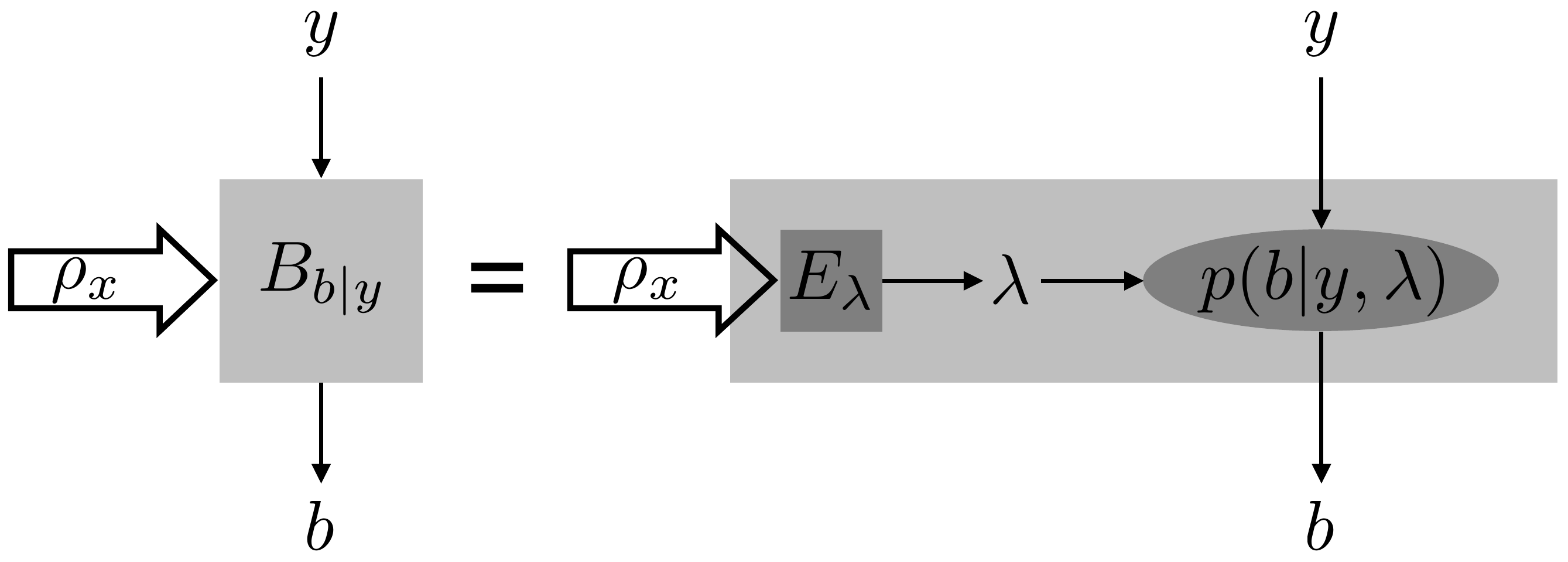}
\caption{Operational definition of joint measurability. Joint measurability asks whether there exists a single POVM $\{E_\lambda\}_\lambda$ from which the original set of POVMs $\{B_{b|y}\}_{b,y}$ can be classically post-processed.} \label{fig:JM}
\end{figure}

A quantum measurement is modelled as a positive operator-valued measure (POVM), i.e. a collection $\{B_{b|y}\}_{b}\subset\L(\mathbb{C}_d)$ for which $B_{b|y}\geq0$ and $\sum_b B_{b|y}=\id$. Here the index $y$ labels the choice of measurement and $b$ is the corresponding outcome. A central concept for sets of measurements $\{B_{b|y}\}_{b,y}$ is \emph{joint measurability}. This asks whether there exists some further POVM from which the set can be post-processed, cf. Fig.~\ref{fig:JM}. More formally, we have the following definition:

\begin{definition}[Joint measurability (JM)]
A set of quantum measurements $\{B_{b|y}\}_{y}$ is jointly measurable (JM) if there exists a POVM $\{E_\lambda\}_\lambda$ and probability distributions $\{p(\cdot|y,\lambda)\}_{y\lambda}$ such that
\begin{align}
	B_{b|y}=\sum_\lambda  p(b|y,\lambda) E_\lambda, \quad \quad \forall b,y
\end{align}
A set of measurements which is not jointly measurable is called \emph{incompatible}. The POVM $\{E_\lambda\}_\lambda$ is called a joint or mother POVM.
\end{definition}

This definition of joint measurability, formulated at the level of relations between operators in Hilbert space, is entirely formal. Nonetheless, in the context of quantum theory, our observations are limited to the \emph{outcome statistics of measurements}. To make the definition more operational practical, consider a bipartite correlation experiment involving two parties, Alice and Bob. Bob possesses a device capable of implementing the various measurements $\{B_{b|y}\}_{y}$ based on some independently chosen input $y$. In each round of the experiment, Alice sends a trusted quantum state $\rho_x$ based on some input $x$. Such prepare-and-measure experiments are characterized by a set of conditional probabilities $\{p(b|x,y)\}_{bxy}$ referred to as the correlations in the experiment. Quantum theory predicts that given these states and measurements, $p(b|x,y) = \tr[\rho_x B_{b|y}]$.

If the preparation device of Alice is completely trusted, we can give the practical, but equivalent definition of joint measurability: 
\newline

\textbf{Definition 1$^*$} (Joint measurability (JM), practical)
\emph{A set of quantum measurements $\{B_{b|y}\}_{b,y}$ is jointly measurable (JM) if for \emph{every} set of trusted preparations $\{\rho_x\}_x$ we can write:
\begin{equation}
p(b|\rho_x,y) = \sum_a p_A(a|\rho_x) p_B(b|y,a),
\end{equation}
where $p_B(b|y,a)$ are classical post-processings and one requires $p_A(a|\rho_x)$ to be linear in the second argument. From Fréchet-Riesz representation theorem, it then follows that $p_A(a|\rho_x) = \tr(\rho_x E_a)$ for some POVM $\{E_a\}_a$.}
\newline

This equivalent definition is advantageous as it directly connects to what we can observe, namely outcomes of experiments. Specifically, it asserts that the joint measurability of the set of quantum measurements is confirmed if, for any conceivable set of trusted preparations $\{\rho_x\}_x$, the observed probabilities $p(b|\rho_x,y)$ can be expressed as a combination of the probabilities $p_A(a|\rho_x)$ related to a mother POVM $\{E_a\}_a$ and classical post-processing $p_B(b|y,a)$. However, it is important to note that this definition assumes perfect trust in Alice's device, implying that the device precisely prepares the states $\{\rho_x\}_x$.

Nonetheless, it is easy to see that the requirement of perfectly trusted preparations is not needed to check the incompatibility of various incompatible measurements. For example, consider that the preparations of Alice are noisy, i.e. instead of sending the states $\rho_x$, she sends the states $\eta \rho_x + (1-\eta) \frac{\id}{d}$. This is equivalent to assuming that the measurements performed by Bob are noisy. For any set of $d$-dimensional measurements $\{B_{b|y}\}_{y}$ and $\eta\in[0,1]$, we define the set of white noisy measurements $\{B_{b|y}^\eta\}_{y}$ as
\begin{align}
	B_{b|y}^\eta:= \eta B_{b|y} + (1-\eta) \frac{\id}{d} \tr(B_{b|y}), \; \forall b,y.
\end{align}

For any set of incompatible measurements $\{B_{b|y}^\eta\}_{y}$, the noisy set $\{B_{b|y}^\eta\}_{y}$ is still incompatible up to some level of noise $\eta$. This can be seen by noting that the set of jointly measurable sets of measurements is closed \cite{Reeb2013}. As a standard example, one can consider the pair of noisy spin measurements in the directions $x$ and $z$. The spin measurements are given by $B_{+|1}=|+\rangle\langle +|, B_{-|1}=|-\rangle\langle -|$, $B_{+|2}=|0\rangle\langle 0|$ and $B_{-|2}=|1\rangle\langle 1|$. The resulting noisy measurements $B_{b|y}^\eta$ are jointly measurable for $\eta\leq1/\sqrt{2}$ with the mother POVM $E_{i,j}=\frac{1}{4}(\openone + \eta(i\sigma_x+j\sigma_z))$, where $i,j\in\{-1,1\}$, and otherwise incompatible \cite{Busch1986}. The post-processing is simply marginalisation, i.e. summing over the index $j$ gives the noisy $x$ measurement and summing over $i$ gives the noisy $z$ measurement. The critical value of $\eta$ at which a set of measurements $\{B_{b|y}^\eta\}_{y}$ becomes jointly measurable serves as a simple metric for noise robustness.

This raises the question of whether we can get a stronger certificate of joint measurability, by lifting the level of trust on Alice's preparations.

\section{Device-independent certification of joint measurability}

\begin{figure*}
\centering
\includegraphics[width= 0.9\linewidth, clip]{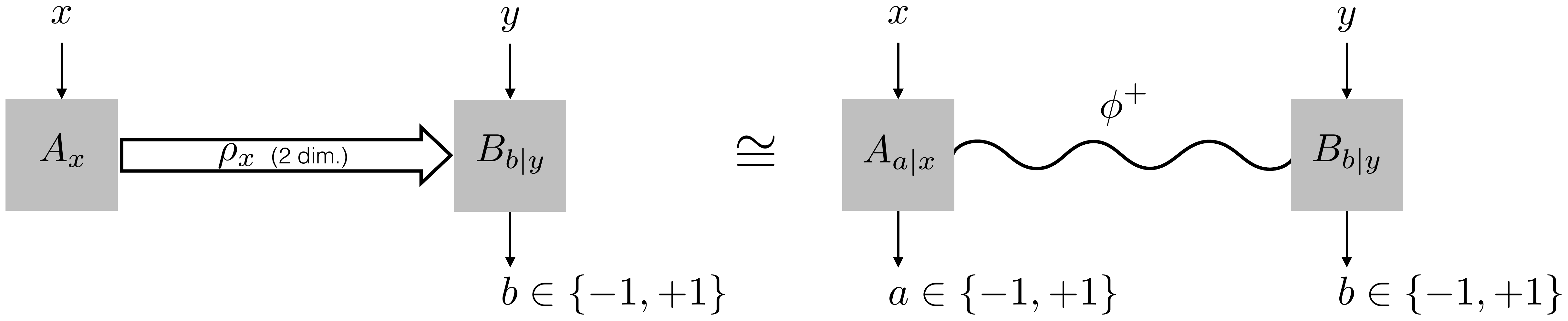}
\caption{Left: The prepare-and-measure scenario for a qubit message and a dichotomic measurement. Right: The Bell scenario where Alice and Bob share a $\phi^+$ state and perform dichotomic measurements. The set of measurements that violate the 2-dimensional classical model and the local hidden variable model, respectively, are equivalent.} \label{fig:PM_Bell}
\end{figure*}

If Alice's preparations $\rho_x$ are completely untrusted, can the parties certify the incompatibility of a set of measurements $\{B_{b|y}\}_y$? The usual approach to analysing the prepare-and-measure scenario (semi-)device-independently, assumes some limitation on the kind of states Alice can send, or, equivalently, the systems supported by the channel connecting the parties \cite{gallego14}. Indeed, if the message $\rho_x$ can encode perfect information about the input $x$, the parties can reproduce any correlations $p$. To have a non-trivial scenario, we can limit the information that $\rho_x$ can contain about $x$. A natural and common approach is to limit the dimension $d$ of the Hilbert space of $\rho_x$, and to allow for shared randomness between the sender and the receiver. Note that for our purposes, this dimension is naturally equal to the dimension of the space that the measurements $\{B_{b|y}\}_y$ act on.

By definition, to certify that a set of measurements $\{B_{b|y}\}_y$ is incompatible entails proving that the correlations cannot be created by performing a single mother POVM plus classical post-processing. Since all correlations in a prepare-and-measure scenario with a single measurement setting can be simulated with classical systems of dimension equaling that of the measurement \cite{Frenkel2015_Weiner}, showing incompatibility with untrusted preparations is thus equivalent to demonstrating that, for \emph{some set of states}, the resulting correlations $p$ cannot be classically simulated. We refer to the correlations in a prepare-and-measure experiment $p$ as \emph{classical} or $\text{PM}_d$ if $p$ can be explained by Alice communicating $d$-dimensional classical systems:

\begin{definition}[$\text{PM}_d$]
A set of probabilities $\{p(b|x,y)\}_{b,x,y}$ can be obtained in a prepare-and-measure scenario with $d$-dimensional classical communication if there exists probability measures $\mu$, $p_A(\cdot|x,\lambda)$, $p_B(\cdot|a,y,\lambda)$, where $a\in{1,\ldots, d}$ such that
\begin{align}
p(b|x,y)= \int \pi(\lambda) \sum_{a=1}^d p_A(a|x,\lambda) p_B(b|a,y,\lambda)d\mu(\lambda)
\end{align}
for all $b,x,y$.
\end{definition}
The smallest $d$ for which given PM correlations do have PM$_d$ model is sometimes called the signalling dimension \cite{DellArno2017}. This leads to the following definition:

\begin{definition}[$\text{PM}_d$-JM]
A set of measurements $\{B_{b|y}\}_{y}$, $B_{b|y}\in L(\mathbb{C}_d)$ is $\text{PM}_d$-JM if, for every set of quantum states $\{\rho_x\}_x$, where $\rho_x\in L(\mathbb{C}_d)$, the set with probabilities given by
\begin{align}
p(b|x,y)=\tr(\rho_x B_{b|y})
\end{align}
is $\text{PM}_d$.
\end{definition}
Note that in the above definitions, we can allow the dimension $d_C$ of the classical simulation to be different from the dimension $d_Q$ of the measurements. This enables us to formulate stronger statements about the incompatibility of a set of measurements when $d_C > d_Q$. 

In Ref.~\cite{degois21_PMJM}, the authors employ a computational polytope approximation method to demonstrate the existence of qubit measurements that are not jointly measurable (JM) but are $\text{PM}_2$-JM. In the next section, we will analytically prove this result by presenting several explicit examples of qubit measurements that are incompatible (not JM) but $\text{PM}_2$-JM. 

Another natural correlation experiment which can be used to certify the incompatiblity of quantum measurements is the Bell scenario. Here, Alice and Bob, who both receive an input, respectively $x$ and $y$, produce outcomes $a$ and $b$, see Fig.~\ref{fig:PM_Bell}. The parties can exchange physical systems before the experiment, but are not allowed to communicate after receiving their inputs. A Bell experiment is described by the conditional probabilities $\{p(a,b|x,y)\}_{a,b,x,y}$, referred to as correlations. 
The most general correlations allowed by quantum theory, are obtained by Alice and Bob performing quantum measurements, respectively $\{A_{a|x}\}_x$ and $\{B_{b|y}\}_y$ on a shared quantum state $\rho_{AB}$. The correlations are given by the Born rule $p(a,b|x,y) = \tr [\rho_{AB} A_{a|x}\otimes B_{b|y}]$.

If the correlations $p$ can be described by Alice and Bob exchanging classical systems before the experiment, we call $p$ \emph{Bell local}. More precisely, we have the following definition:

\begin{definition}[Bell locality]
A set of probabilities $\{p(a,b|x,y)\}_{a,b,x,y}$ is Bell local if there exists sets of probability measures $\mu$, $p_A(\cdot|x,\lambda)$, and $p(\cdot|y,\lambda)$ such that
\begin{align}
	p(a,b|x,y)=\int \pi(\lambda) p_A(a|x,\lambda) p_B(b|y,\lambda) d\mu(\lambda)\quad
\end{align}
for all $a,b,x,y$.

\end{definition}

This naturally leads to the following definition:

\begin{definition}[Bell-JM and Bell$^{\phi^+}$-JM]
A set of measurements $\{B_{b|y}\}_{y}$, $B_{b|y}\in L(\mathbb{C}_d)$ is Bell-JM if for every quantum state $\rho_{AB}\in L(\mathbb{C}_{d'}\otimes \mathbb{C}_d)$ and every set of measurements $\{A_{a|x}\}_{a,x}$ the set with probabilities given by
\begin{align}
	p(a,b|x,y)=\tr(\rho_{AB} A_{a|x}\otimes B_{b|y})
\end{align}
is Bell local. 

The set of measurements $B_{b|y}\in L(\mathbb{C}_d)$ is Bell$^{\phi^+}$-JM if 
\begin{align}
	p(a,b|x,y)=\tr(\ketbra{\phi^+_d}{\phi^+_d} A_{a|x}\otimes B_{b|y})
\end{align}
is Bell local, where $\ket{\phi^+_d}:=\frac{1}{\sqrt{d}}\sum_{i=0}^{d-1}\ket{ii}$ is the maximally entangled state.
\end{definition}

The relationship of between JM and Bell-JM has been analysed in 
Refs.~\cite{quintino14,uola14},  which present examples of measurements which are not JM but are Bell$^{\phi^+}$-JM in a scenario where Alice's measurements are dichotomic.
Later, Ref.~\cite{quintino15b} shows an example of measurements which are not JM but are Bell-JM in a scenario where Alice's measurements are dichotomic. Finally, Ref.~\cite{hirsch19_JMBELL} and \cite{bene18_BELLJM} showed that there exists sets of qubit measurements which are not JM, but they are Bell-JM, regardless of what Alice measures.

\section{Main results}

In this section, we establish and demonstrate results pertaining to the certification of joint measurability in both prepare-and-measure and Bell correlations, with a specific focus on qubit scenarios. To commence, we have the following theorem~\cite{renner23_PM}:

\begin{theorem} [PM$_4$-JM for qubits]
Any set of qubit measurements $\{B_{b|y}\}_{y}$ is PM$_4$-JM.
\end{theorem}
\begin{proof}
	Ref.~\cite{renner23_PM} shows that all qubit quantum PM strategies can be perfectly simulated classically with two bits of communication. In other words, for qubits, all sets with quantum probabilities $p(b|x,y)=\tr(\rho_x B_{b|y})$ are PM$_4$.
\end{proof}
Consequently, the following theorem represents the strongest form of incompatibility in prepare-and-measure qubit systems concerning classical simulability:
\begin{theorem} [PM$_3$-JM for qubits]
There exists sets of qubit measurements $\{B_{b|y}\}_{y}$ which are not PM$_3$-JM.
\end{theorem}
\begin{proof}
	Ref.~\cite{renner23_PM} presents an explicit example of a qubit strategy in a PM scenario which cannot be obtained via trit classical strategies. Hence, the qubit measurements used in this work are not PM$_3$-JM. In the example of Ref.~\cite{renner23_PM}, Alice's states are the 6 eigenvectors of Pauli matrices and Bob has 24 different projective measurements, which are chosen such that their representation on the Bloch sphere is the vertices of a snub cube.
\end{proof}

Our subsequent theorem strongly builds upon the results and methodologies outlined in Ref.~\cite{divianszky23_PM_Bell_Grothendieck}, which establishes a one-to-one connection between specific full-correlation dichotomic Bell scenarios and dichotomic prepare-and-measure scenarios.
\begin{theorem}\label{thm:main} [PM$_2$-JM and Bell$^{\phi^+}$-JM for qubits] Let $\{B_{b|y}^{\eta_y}\}_{y}$ be a set of qubit noisy projective measurements, that is, $B_{b|y} \in L(\mathbb{C}_2)$ and we can write $B_{0|y}^{\eta_y}=\eta_y\ketbra{\phi_y}+(1-\eta_y)\frac{\id}{2}$, for some qubit pure state $\ket{\phi_y}$ and $\eta_y\in[0,1]$.
	The set of noisy measurements $\{B_{b|y}^{\eta_y}\}_{y}$ is Bell$^{\phi^+}$-JM if and only if it is  $\text{PM}_2$-JM.
\end{theorem}
	The proof of Thm.~\ref{thm:main} is presented in Appendix~\ref{sec:proof}. We note that the measurements used in the above Theorem are exactly the ones, whose POVM elements $\{B_{b|y}^{\eta_y}\}_{b,y}$ admit a Bloch vector representation similar to quantum states, that is, we can write $B_{0|y}^{\eta_y}:=\frac{1}{2}(\openone+\vec{b}_y\cdot \vec{\sigma})$, where $\vec{b}_y\in\mathbb{R}^3$ is a vector with Euclidean norm $\norm{\vec{b}_y}=\eta_y\leq 1$, and $\vec{b}_y\cdot \vec{\sigma} := b_y^{1} \sigma_x +  b_y^{2}  \sigma_y + b_y^{3} \sigma_z.$
 
We notice that the proof is constructive in the follwoing sense: let $\{B_{b|y}\}$ be a set of measurements which are not PM$_2$, i.e. there exist a set of states $\{\rho_x\}_x$ such that $\tr(\rho_x B_{b|y})\notin \text{PM}_2$. Then,the behaviour $p(ab|xy)=\tr(\phi_2^+ A_{a|x}\otimes B_{b|y})\notin \text{Bell}$, where $A_{0|x}:=\rho_x$ and $A_{1|x}:=\id - \rho_x$. Conversely, if for the set of measurements $\{A_{a|x}\}_{x}$, we have $p(ab|xy)=\tr(\phi_2^+ A_{a|x}\otimes B_{b|y})\notin \text{Bell}$ for some $\{B_{b|y}\}_y$, it holds that $\tr(\rho_x B_{b|y})\notin \text{PM}_2$,  with the states $\rho_x=A_{0|x}$. We now discuss some consequences of Thm.~\ref{thm:main}.
\begin{corollary} \label{coro:2}
	Let $\left\{X^\eta, Y^\eta, Z^\eta\right\}$ be the three noisy Pauli measurements. This set is JM iff $\eta\leq\frac{1}{\sqrt{3}}$ and PM$_2$-JM iff $\eta\leq \frac{1}{\sqrt{2}}$. Hence, this yields a simple example that $PM_2-JM$ does not imply $JM$ even for qubits. 
\end{corollary}
\begin{proof}
We start by noticing that, for the CHSH Bell inequality, if Alice and Bob share a maximally entangled state and Alice has $\eta$ white noisy Pauli measurements, it is possible to obtain a CHSH violation if and only if $\eta>1/\sqrt{2}$. To finish the proof, we simply notice that in dichotomic full correlations scenarios where one party has three measurements, the only non-trivial Bell inequality is the CHSH~\cite{Avis2006}.
\end{proof}
\begin{corollary}\label{coro:2}
	Let $\left\{\Pi_{b|y}^\eta\right\}$ be the set of all noisy projective qubit measurements. This set is $PM_2-JM$ iff $\eta\leq \frac{1}{K_3}$, where $K_3$ is the real Grothendieck constant of order three, which is known to respect~\cite{designolle2023_Grothedieck} $0.6875 \leq \frac{1}{K_3}< 0.6961$.
	
	Let $\left\{\Pi_{b|y}^{2,\eta}\right\}$ be the set of all qubit noisy projective planar measurements, i.e. measrurements whose Bloch vectors lie in a plane. This set is $PM_2-JM$ iff $\eta\leq \frac{1}{K_2}$, where $K_2=\sqrt{2}$ is the real Grothendieck constant of order two.
\end{corollary}
\begin{proof}
 The two-qubit isotropic state $\phi^+_\eta:=\eta\ketbra{\phi^+_2} + (1-\eta)\frac{\id_4}{4}$ violates a Bell inequality with projective measurements if and only if $\eta>\frac{1}{K_3}$~\cite{tsirelson93,acin06}. That is, if $\eta>\frac{1}{K_3}$, there exists measurements such that $\tr(\phi^+_\eta\, \Pi^A_{a|x}\otimes \Pi_{b|y})$ violates a Bell inequality. To finish the first part of the proof, we just notice that $\tr(\phi^+_\eta\, \Pi^A_{a|x}\otimes \Pi^B_{b|y})=\tr(\phi^+\, \Pi^A_{a|x}\otimes \Pi^\eta_{b|y})$. 

To prove the result for planar measurement, we notice that the two-qubit maximally entangled state violates a Bell inequality with projective planar measurements if and only if $\eta>\frac{1}{K_2}$~\cite{acin06}.
\end{proof}
We notice that, although stated differently, corollary~\ref{coro:2} is equivalent to Lemma 2 and Lemma 3 from Ref.~\cite{divianszky23_PM_Bell_Grothendieck}, and that the relationship between the Grothendieck constant with Bell and PM scenarios has been extensively discussed in Refs.~\cite{tsirelson93,acin06,vertesi08,hirsch16,Divianszky2017,designolle2023_Grothedieck,divianszky23_PM_Bell_Grothendieck}. 

We notice that corollary~\ref{coro:2} provides various examples of qubit measurements which are not JM but are $\text{PM}_2$-JM. For example, any set of $\eta$-noisy planar measurements with $\eta=1/\sqrt{2}$ is PM$_2$-JM. However, if a set of $\eta=1/\sqrt{2}$ noisy planar measurements includes two Pauli measurements and a non-Pauli one, i.e. one measurement has a Bloch vector not parallel to the other two, it is not JM.  This provides several examples of incompatible measurements which cannot be certified in a PM scenario. Also, since $0.6875 \leq \frac{1}{K_3}$,  any set of $\eta$-noisy qubit measurements with $\eta\leq 0.6875$ cannot be certified in PM scenario.

By using the connection between PM$_2$-JM and Bell$^{\phi^+}$-JM together with a known characterization of joint measurability \cite{Busch1986}, we get the following operator-inequality based connection between the three concepts.

\begin{theorem}\label{coro:3}
	Let $\{B_{b|y}^\eta\}_y$ be a pair of qubit white noisy projective measurements, that is $y\in\{1,2\}$. The set $\{B_{b|y}^\eta\}_y$ is JM iff it is Bell-JM, iff it is PM$_2$-JM, and iff $\norm{B_0^\eta+B_1^\eta}+\norm{B_0^\eta-B_1^\eta}\leq 2$, where $B_y:=B_{0|y}-B_{1|y}$.
\end{theorem} 
 \begin{proof} The claim on JM follows from the known criterion of Busch \cite{Busch1986} together with a suitable argument using the CHSH operator, cf. Appendix~\ref{sec:proof_coro3}, and Theorem~\ref{thm:main}.\end{proof}

\section{Discussions and open questions}

A crucial question in quantum communications is under which trust assumptions we can certify quantum properties. A common and natural approach is to limit the dimension of the Hilbert space, but various other alternative assumptions have been studied \cite{pauwels2024information}. This work presents such analysis for measurement resources in a prepare-and-measure setting, concentrating on a dimension-bounded generalisation of joint measurability. Here, we perform a case study for qubit measurements. These examples showcase the impact of trust assumptions on incompatibility witnesses (cf. Corollary~\ref{coro:2}).

The present work suggests several interesting directions for future research. For example, a major open question is to establish a hierarchy between the set of PM$_d$ and Bell-local measurements beyond the simple case of qubits and the $\phi^+$ state. More precisely, given a set of arbitrary $d$-dimensional measurements, does  $\{B_{b|y}\}_y\in$ PM$_d$ imply $\{B_{b|y}\}_y\in$ Bell-Local? Or, does $\{B_{b|y}\}_y\in$ Bell-Local imply $\{B_{b|y}\}_y\in$ PM$_d$?  The intimate connection between the Bell and prepare-and-measure scenarios \cite{Bennett1992,Catani2023} might naively suggest the existence of such a hierarchy. Notably, an invertible map exists between correlations in a Bell scenario and a subset of prepare-and-measure scenarios with preparations satisfying appropriate equivalences \cite{Wright2023}, see also Ref.~\cite{kiukas17}. However, the role of classical simulation dimensions in this relationship remains unclear.
 
In the context of qubits, the set of all possible POVMs is PM$_4$-JM, implying that qubit measurement incompatibility cannot be certified in a scenario where Alice can transmit two classical bits to Bob. The extension of this concept to qudits prompts intriguing questions. For instance, consider the set of all qutrit measurements $\{B_{b|y}\}_{b,y}$. Is there a classical dimension $d\in\mathbb{N}$ such that $\{B_{b|y}\}_{b,y}$ is PM$_d$-JM?

For qubits, does Bell$^{\phi^+_2}$-JM imply Bell-JM in general? This question was already addressed in various previous references, e.g.~\cite{quintino14,uola14,quintino15b,hirsch19_JMBELL,bene18_BELLJM}, but it seems very hard to solve. 

Another interesting direction is to generalise the definition of PM$_d$-JM, by asking for the existence of a simulation with dimension $d$ only \emph{on average} (cf. \cite{Pauwels2022} Appendix C), i.e. can the correlations be simulated by probabilistically choosing a simulation $\lambda$ of dimension $d_\lambda$ with probability $p_\lambda$ such that $d = \sum_\lambda p_\lambda d_\lambda$? This promotes $d$, as a measure of joint measurability, to a continuous quantity~\footnote{While it was proven in \cite{renner23_PM} that bit simulation strategies for the set of all qubit correlations have measure zero, for trit simulations this is still an open problem.}. A further interesting question is to investigate how PM$_d$-JM relates to signaling dimension of channels \cite{Chitambar2021}, classical simulation of temporal correlations \cite{Hoffmann2018}, and how it generalizes to GPTs \cite{Heinosaari2020,dallarno2023signaling}.

More broadly, one could study the impact of different trust assumptions on incompatibility and other measurements resources. For example, can one formulate an analogue of energy bounds for measurements and how does a distrust  assumption affect measurement simulability \cite{Oszmaniec2017,Guerini2017,Ioannou2022}.

\section{Acknowledgments}
S.E. and R.U. are grateful for the financial support from the Swiss National Science Foundation (Ambizione PZ00P2- 202179). J.P. acknowledges support from NCCR-SwissMAP.

\newpage
\appendix
\onecolumngrid
	\section{Proof of Thm.~\ref{thm:main}} \label{sec:proof}
 
The idea of the proof is to show an explicit bijective map between qubit PM strategies and Bell strategies with a pair of maximally entangled qubit pairs. Before presenting this construction, we establish some notation. We define the (full) correlator in a Bell scenario
\begin{align}
	C_{xy}:=p(a=b|x,y)-p(a\neq b|x,y) \, .
\end{align}
In a dichotomic Bell scenario with measurements $\{A_{a|x}\}_{a,x}$ and $\{B_{b|y}\}_{b,y}$, we define the observables
\begin{align}
&A_x:= A_{0|x}-A_{1|x}\, , \\
&B_y:= B_{0|y}-B_{1|y} \, .
\end{align}
With these definitions, if the parties share the state $\rho$, the correlator is given by
\begin{align}
C^{\rho}_{xy}=\tr(\rho A_x\otimes B_y) \, .
\end{align}
Furthermore, for $\rho =\ketbra{\phi^+_d}{\phi^+_d}$ we have the identity
\begin{align}
C^{\phi^+}_{xy}=\frac{1}{d}\tr(A_x^T B_y)\, .
\end{align}
Also, if the observable $A_x$ and $B_y$ respect $\tr(A_x)=\tr(B_y)=0$, if we know the value of $P_{xy}:=\tr(\rho_x B_y)$, we can always recover the probabilities via the identity $p(a,b|x,y)=\frac{1+(-1)^{a \oplus b} C^{\phi^+}_{xy}}{4}$, where $a\oplus b$ is the addition modulo two. A proof of this identity follows from direct calculation.

Similarly, in the dichotomic PM-scenario we define the (single) correlator as
\begin{align}
	P_{xy}:=p(b=0|x,y)-p(b=1|x,y) \, .
\end{align}
If Alice sends the states $
\rho_x$ and Bob performs measurements $B_{b|y}$ the correlator is given by
\begin{align}
	P_{xy}:=\tr(\rho_x B_y) \, .
\end{align}
Notice that since $p(b=0|x,y)+p(b=1|x,y)=1$, if we know the value of $P_{xy}:=\tr(\rho_x B_y)$, we can always recover the probabilities via the identity $p(b=0|x,y)=\frac{P_{xy}+1}{2}$.

We are now ready to state Lemma~\ref{lemma:1}, which employs ideas from Lemma II.1, Lemma II.2, and Lemma II.3 from Ref.~\cite{Divianszky2017}. 
\begin{lemma} \label{lemma:1}
Let $P_{xy}:=\tr(\rho_x B_y)$ be qubit correlations in a dichotomic PM scenario where Alice has $N_A$ inputs and Bob has $N_B$ inputs with\footnote{Assuming $\tr(B_y)=0$ is equivalent to saying that Bob's POVMs are given by $B_{0|y}=\eta_y \ketbra{\phi_y}+(1-\eta_y)\frac{\id}{2}$ for some state 
$\ket{\phi_y}$ and $\eta_y\in [0,1]$.} $\tr(B_y)=0$. Let us define the new correlator $P'_{xy}$ in a scenario where Alice includes the states given by $\id-\rho_x$. That is, Alice and Bob have $2N_A$ and $N_B$ inputs respectively, and 
\begin{align}
P'_{xy}:=&\tr(\rho_x B_y)=P_{xy}, \quad \quad &\forall x\in\{1,\ldots, N_A\}, \forall y\in\{1,\ldots,N_B\}\\
P'_{xy}:=&\tr((\id-\rho_{(x-N_A)}) B_y)=-P_{(x-N_A)y},  &\forall x\in\{N_A+1,\ldots, 2N_A\}, \forall y\in\{1,\ldots,N_B\}.
\end{align}
We now define the POVMs 
\begin{align} \label{eq:AlicePOVM}
    &A_{0|x}:=\rho_x^T \, , \\
    &A_{1|x}:=\id-\rho_x^T \,.
\end{align}

The probabilities $p(a,b|x,y)=\tr(\ketbra{\phi^+_d}{\phi^+_d} A_{a|x}\otimes B_{b|y})$ are Bell local if and only if the probabilities associated to $P'_{xy}:=\tr(\rho_x B_y)$ are in PM$_2$.
\end{lemma}
\begin{proof}
    We start the proof by defining the Bell correlator,
\begin{align}
C'_{xy}:=&\tr(\ketbra{\phi^+_2} A_x\otimes B_y) \quad \quad &\forall x\in\{1,\ldots, N_A\},\forall y\in\{1,\ldots,N_B\}\\
C'_{xy}:=&-\tr(\ketbra{\phi^+_2} A_x\otimes B_y) \quad \quad &\forall x\in\{N_A+1,\ldots, 2N_A\},\forall y\in\{1,\ldots,N_B\} \,,
\end{align}
where $A_x=2\rho_x^T-\id$ is the observable associated to the measurements $A_{a|x}$ defined in Eq.~\eqref{eq:AlicePOVM}.

Our next step is to show that $C'_{xy}=P'_{xy}$. This will show that, if the probabilities of $P'_{xy}$ in the PM scenario can be obtained with qubits, we can construct a behaviour $p(a,b|x,y)=\tr(\ketbra{\phi^+_d}{\phi^+_d} A_{a|x}\otimes B_{b|y})$ with correlator $C'_{xy}$. Also, given the behaviour $p(a,b|x,y)=\tr(\ketbra{\phi^+_d}{\phi^+_d} A_{a|x}\otimes B_{b|y})$, we can construct a qubit PM strategy respecting $C'_{xy}=P'_{xy}$ for every $x$ and $y$. For this, it is enough to notice that
\begin{align}
C'_{xy}:=&\tr(\ketbra{\phi^+_2} A_x\otimes B_y) \quad \quad &\forall x\in\{1,\ldots, N_A\},\forall y\in\{1,\ldots,N_B\}\\
=&\frac{1}{2}\tr( A_x^T\otimes B_y)\\
=&\frac{1}{2}\Big(2\tr( \rho_x B_y)-\tr(B_y)\Big)\\
=&\tr( \rho_x B_y)\\
=&P_{xy} \\
=&P_{xy}' \,,
\end{align}
where the last inequality holds because $ x\in\{1,\ldots, N_A\},\forall y\in\{1,\ldots,N_B\}$.
Analogous calculation with $x\in\{N_A+1,\ldots, 2N_A\}$ shows that $C'_{xy}=P'_{xy}$ for all admissible values of $x$ and $y$.

We now prove that if $P'_{xy}\in\text{PM}_2$, then $C'_{xy}$ is Bell local. The intuition of the proof is that since $P'_{xy}\in\text{PM}_2$, the behaviour of $P'_{xy}$ may be obtained in a quantum scenario (potentially with shared randomness) where the measurements performed by Bob are jointly measurable. Hence, the corresponding quantum behaviour $C'_{xy}=P'_{xy}$ is necessarily Bell local. Below, we describe this proof in more details.\\
If $P'_{xy}\in\text{PM}_2$, then 
for any fixed $\lambda$, 
by definition $p(b|x,y,\lambda)=\sum_{a=0}^1 p_A(a|x,\lambda) p_B(b|y,a,\lambda)$. 
Every such strategy can be viewed as a quantum strategy where Alice sends a qubit which is diagonal in the computational basis, that is, $\rho_{x,\lambda}^c=\sum_a  p_A(a|x,\lambda)\ketbra{a}$ and Bob measures in the computational basis followed by classical post-processing, i.e., $M_{b|y,\lambda}^c:=\sum_a p_B(b|y,a,\lambda) \ketbra{a}$. 
Hence, $P'_{xy,\lambda}\in\text{PM}_2$ can be written as 
\begin{align}
    P'_{xy,\lambda} = \tr(\rho_{x,\lambda}^c M_{y,\lambda}^c)\, ,
\end{align}
where the operators $\rho_{x,\lambda}^c$ and $M_{b|y,\lambda}^c$ are diagonal in the same basis. From the construction we presented in Eq.~\eqref{eq:AlicePOVM}, we see that for each $\lambda$ the correlator $C'_{xy,\lambda}$ is obtained in a scenario where Alice's measurements are given by
\begin{align}
    &A_{0|x,\lambda}^c:= {\rho_{x,\lambda}^c}^T~, \\
    &A_{0|x,\lambda}^c:=  \openone - {\rho_{x,\lambda}^c}^T
\end{align}
Since all POVM elements of Alice's measurements are diagonal in the same basis, her set of measurements is jointly measurable and cannot lead to Bell nonlocal correlations, regardless of the shared state or the measurements performed by Bob \cite{fine82,wolf09,quintino14,uola14}, hence for each fixed $\lambda$, the correlator $C'_{xy,\lambda}$ is Bell local. Since the set of Bell local correlations is convex, convex combination of local correlators are also Bell local, that is, $\sum_\lambda \pi(\lambda) C'_{xy,\lambda}$ is necessarily Bell local.

We now prove that if the behaviour with probabilities $p(a,b|x,y)=\tr(\ketbra{\phi^+_d}{\phi^+_d} A_{a|x}\otimes B_{b|y})$ is Bell local, then $P'_{xy}\in\text{PM}_2$. For that, we will show that, if $P'_{xy}\notin\text{PM}_2$, then the correlator $C_{xy}'$ violates a correlator Bell inequality.\\

If $P'_{xy}\notin\text{PM}_2$, there exists a hyperplane that separates it from the classical set, that is, there exists real numbers $M_{xy}$ such that, 
\begin{align}
    \sum_{xy} M_{xy} P'_{xy} = Q \, ,
\end{align}
but, for all $P^{\text{PM}_2}_{xy}\in\text{PM}_2$\, ,
\begin{align}
    \max_{P^{\text{PM}_2}_{xy}\in \text{PM}_2} \Big[\sum_{xy} M_{xy} P^{\text{PM}_2}_{xy} \Big] = L<Q \, .
\end{align}

Since, $P'_{xy}=C'_{xy}$, it follows 
\begin{align}
    \sum_{xy} M_{xy} C'_{xy} = Q \,.
\end{align}
Finally, we have to show that all local Bell correlations $C^\text{LHV}_{xy}$ respect 
\begin{align}
   \max_{C^\text{LHV}_{xy}\in \text{LHV}} \Big[\sum_{xy} M_{xy} C^\text{LHV}_{xy} \Big] = L\, ,
\end{align}
that is, we have to show that $\sum_{xy} M_{xy} C_{xy}\leq L$ is a Bell inequality. This is ensured by Lemma II.1 of Ref.~\cite{Divianszky2017}, which shows that if 
  $\max_{P^{\text{PM}_2}_{xy}\in \text{PM}_2} \Big[\sum_{xy} M'_{xy} P^{\text{PM}_2}_{xy} \Big] = L$, and the coefficients' $M'_{xy}$ can be written $M'_{xy}=M_{xy}$ for $x\in\{1,\ldots N_A\}$ and $M'_{xy}=-M_{xy}$ for $x\in\{N+A+1,\ldots 2N_A\}$, it holds that $\max_{C^\text{LHV}_{xy}\in \text{LHV}} \Big[\sum_{xy} M_{xy} C^\text{LHV}_{xy} \Big] = L$. And, since our correlator $P'_{xy}$ respects $P'_{xy}=P_{xy}$ for $x\in\{1,\ldots N_A\}$ and $P'_{xy}=-P_{xy}$ for $x\in\{N+A+1,\ldots 2N_A\}$, we can always find coefficients $M'_{xy}$ meeting the hypothesis of Lemma II.1 of Ref.~\cite{Divianszky2017}.

\end{proof}

Using Lemma~\ref{lemma:1}, the proof of Thm.~\ref{thm:main} follows immediately. Let $B_{b|y}$ be a set of qubit measurements which can be written as $B_{0|y}=\eta_y\ketbra{\phi_y}+(1-\eta_y)\frac{\id}{2}$, for some qubit pure state $\ket{\phi_y}$ and $\eta_y\in[0,1]$.

If there exists qubit states $\rho_x$ such that $p(b|x,y)=\tr(\rho_x B_{b|y})\notin PM_2$, then there exist measurements $A_{a|x}$ such that $p(a,b|x,y)=\Tr(A_{a|x}\otimes B_{b|y} \phi^+_2)$ is not Bell local. 

\section{Proof of Theorem~\ref{coro:3}} \label{sec:proof_coro3}
\begin{definition}
Let $A_0,A_1,B_0,B_1\in L(\mathbb{C}_d)$ be self-adjoint operators with eigenvalues contained in $[-1,1]$.
The CHSH operator is defined as
\begin{align}
	\text{CHSH}:=A_0\otimes B_0+A_1\otimes B_0+A_0\otimes B_1-A_1\otimes B_1 \, .
\end{align}
\end{definition}
\begin{lemma}
	The CHSH operator satisfies the norm inequality $\norm{CHSH}\leq \norm{B_0+B_1}+\norm{B_0-B_1}$.
\end{lemma}
\begin{proof}
\begin{align}
	\norm{CHSH}=&\norm{A_0\otimes( B_0+B_1) + A_1\otimes( B_0-B_1)} \\
		\leq &\norm{A_0\otimes( B_0+B_1)} + \norm{A_1\otimes( B_0-B_1)} \\
		=& \norm{A_0}\norm{B_0+B_1} + \norm{A_1}\norm{ B_0-B_1}\\
		\leq& \norm{B_0+B_1} + \norm{ B_0-B_1} \, .
\end{align}
\end{proof}
This shows that if Bob performs measurements which satisfy $\norm{B_0+B_1}+\norm{B_0-B_1}\leq2$, Alice and Bob cannot violate the CHSH inequality, regardless of the shared state, i.e. the measurements are Bell-JM. The following Theorem shows that the above upper bound can be reached with $\ket{\phi^+}$.
\begin{theorem}
	Let $B_0,B_1\in L(\mathbb{C}_2)$ be qubit self-adjoint operators with eigenvalues in $[-1,1]$ and $\tr(B_y)=0$. There exist self-adjoint operators $A_0,A_1\in L(\mathbb{C}_2)$ with eigenvalues contained in $\{-1,1\}$  and a quantum state $\ket{\psi}\in\mathbb{C}_2\otimes\mathbb{C}_2$ such that 
\begin{align}
\bra{\psi} CHSH \ket{\psi}= \norm{B_0+B_1}+\norm{B_0-B_1} \, .
\end{align}
Moreover, the state $\ket{\psi}\in\mathbb{C}_2\otimes\mathbb{C}_2$ may be taken as the maximally entangled state $\ket{\phi^+}:=\frac{1}{\sqrt{2}}\sum_{i=1}^2 \ket{ii}$.
\end{theorem}
\begin{proof}
First, notice that, for any operators $A_0,A_1,B_0,B_1\in L(\mathbb{C}_2)$, we have that
\begin{align}
\bra{\phi^+} CHSH \ket{\phi^+}=\frac{\tr\Big(A_0^T\otimes( B_0+B_1)\Big) + \tr\Big(A_1^T\otimes( B_0-B_1)\Big)}{2} \, .
\end{align}	
By definition of operator norm, there exists normalised vectors $\ket{\psi_0'},\ket{\psi_1'}\in L(\mathbb{C}_2)$ such that
\begin{align}
\abs{\bra{\psi_0'}B_0+B_1\ket{\psi_0'}}=&\norm{B_0+B_1} \, , \\
\abs{\bra{\psi_1'}B_0-B_1\ket{\psi_{1}'}}=&\norm{B_0-B_1} \, .
\end{align}
Since $\tr(B_y)=0$, there exists 
$\ket{\psi_0},\ket{\psi_1}\in L(\mathbb{C}_2)$ such that
\begin{align}
\bra{\psi_0}B_0+B_1\ket{\psi_0}=&\norm{B_0+B_1} \,  ,\\
\bra{\psi_1}B_0-B_1\ket{\psi_{1}}=&\norm{B_0-B_1} \, .
\end{align}
Indeed, if $\bra{\psi_0'}B_0+B_1\ket{\psi_0'}\geq0$, take $\ket{\psi_0}:=\ket{\psi_0'}$. If $\bra{\psi_0'}B_0+B_1\ket{\psi_0'}<0$, take $\ket{\psi_0}$ as a vector which satisfies $\ketbra{\psi_0}=\id - \ketbra{\psi_0'}$, that is, $\ket{\psi_0}$ is the universal NOT of the state $\ket{\psi_0'}$. It follows that
\begin{align}
\bra{\psi_0}B_0+B_1\ket{\psi_0}&=-\tr(\ketbra{\psi_0'}(B_0+B_1)) \nonumber\\
&=-\bra{\psi_0'}B_0+B_1\ket{\psi_0'}
\end{align}
is a positive number. An analogous argument works for $\ket{\psi_1}$.

We now define $A_x=:2\ketbra{\psi_x}^T-\id$. We have
\begin{align}
\bra{\phi^+} CHSH \ket{\phi^+}&=\frac{\tr\Big(A_0^T\otimes( B_0+B_1)\Big) + \tr\Big(A_1^T\otimes( B_0-B_1)\Big)}{d}\\
&=\frac{2}{d}\left(\bra{\psi_0}B_0+B_1\ket{\psi_0}+\bra{\psi_1}B_0-B_1\ket{\psi_1}\right)-\Big(\tr(B_0)+\tr(B_1)+\tr(B_0)-\tr(B_1)\Big)\\
&=\frac{2}{d}\left(\bra{\psi_0}B_0+B_1\ket{\psi_0}+\bra{\psi_1}B_0-B_1\ket{\psi_1}\right)-2\tr(B_0) \, .
\end{align}
Since $\tr(B_0)=0$, we have
\begin{align}
\bra{\phi^+} CHSH \ket{\phi^+}
&=\frac{2}{2}\left(\bra{\psi_0}B_0+B_1\ket{\psi_0}+\bra{\psi_1}B_0-B_1\ket{\psi_1}\right) \, ,\\
&= \norm{B_0+B_1}+\norm{B_0-B_1} \, .
\end{align}
\end{proof}
\bibliographystyle{apsrev4-2.bst}

\bibliography{0_MTQ_bib.bib}
\end{document}